\documentclass[10pt,journal,cspaper,compsoc]{IEEEtran}
\ifCLASSOPTIONcompsoc
\else
\fi
\usepackage{graphicx}
\usepackage{graphics}
\usepackage{lscape}

\ifCLASSINFOpdf
\else
\fi
\usepackage{algorithm}
\usepackage{algorithmic}
\usepackage{arydshln}
\usepackage{url}

\newtheorem{theorem}{Theorem}
\newtheorem{lemma}{Lemma}

\begin{document}

\newcommand{\XOR}{\otimes}
\newcommand{\BXOR}{\bar{\otimes}}

%
\title{Evolution and Controllability of Cancer Networks: a Boolean Perspective}
%
%
%
%

\author{Sriganesh Srihari,
        Venkatesh Raman,
        Hon Wai Leong
        and Mark A. Ragan
\IEEEcompsocitemizethanks
{\IEEEcompsocthanksitem Sriganesh Srihari is with the Institute for Molecular Bioscience,
The University of Queensland, QLD 4072, Australia.\protect\\
E-mail: s.srihari@uq.edu.au
\IEEEcompsocthanksitem Venkatesh Raman is with The Institute of Mathematical Sciences, CIT Campus Taramani, Chennai 600113, India.\protect\\
E-mail: vraman@imsc.res.in
\IEEEcompsocthanksitem Hon Wai Leong is with the Department of Computer Science, National University of Singapore, Singapore 117590. \protect\\
E-mail: leonghw@comp.nus.edu.sg
\IEEEcompsocthanksitem Mark A. Ragan is with the Institute for Molecular Bioscience, The University of Queensland, QLD 4072, Australia. \protect\\
E-mail: m.ragan@uq.edu.au
}
\thanks{}}

%
%

\markboth{Journal of \LaTeX\ Class Files,~Vol.~6, No.~1, January~2007}%
{Shell \MakeLowercase{\textit{et al.}}: Bare Demo of IEEEtran.cls for Computer Society Journals}
%


\IEEEcompsoctitleabstractindextext{%
\begin{abstract}
Cancer forms a robust system capable of maintaining stable functioning (cell sustenance and proliferation) despite perturbations.
Cancer progresses as stages over time typically with increasing aggressiveness and worsening prognosis. Characterizing these stages and identifying the genes driving transitions 
between them is critical to understand cancer progression and to develop effective anti-cancer therapies.
In this work, we propose a novel model for the `cancer system' as a Boolean state space in which a Boolean network, built from protein-interaction and gene-expression data from 
different stages of cancer, transits between Boolean satisfiability states by ``editing" interactions and ``flipping" genes. 
Edits reflect rewiring of the PPI network while flipping of genes reflect activation or silencing of genes between stages.
We formulate a minimization problem {\sc min flip} to identify these genes driving the transitions.
The application of our model (called BoolSpace) on three case studies -- pancreatic and breast tumours in human and post spinal-cord injury in rats -- 
reveals valuable insights into the phenomenon of cancer progression:
(i) interactions involved in core cell-cycle and DNA-damage repair pathways are significantly rewired in tumours, indicating significant impact to key genome-stabilizing mechanisms;
(ii) several of the genes flipped are serine/threonine kinases which act as biological switches, reflecting cellular switching mechanisms between stages; and
(iii) different sets of genes are flipped during the initial and final stages indicating a pattern to tumour progression.
Based on these results, we hypothesize that robustness of cancer partly stems from ``passing of the baton" between genes at different stages -- 
genes from different biological processes and/or cellular components are involved in different stages of tumour progression thereby allowing tumour cells to evade targeted therapy, and
therefore an effective therapy should target a ``cover set" of these genes.
A C/C++ implementation of BoolSpace is freely available at: \url{http://www.bioinformatics.org.au/tools-data}

\end{abstract}

}

\maketitle

\IEEEdisplaynotcompsoctitleabstractindextext

%
\IEEEpeerreviewmaketitle

\section{Introduction}
A dynamical system is \emph{controllable} if it can be driven from an initial state to a desired state within finite time by application of suitable inputs~\cite{LiuY2011}.
For example, a car is controllable as it can be moved at a desired speed and direction by the manipulation of pedals and steering wheel.
The factors that contribute to the \emph{controllability} of the system can be assembled in the form a network, which in this example is the network of components such as circuits, engine,
wheels, etc. of the car. This prompts the study of \emph{structural controllability of networks} wherein we attempt to identify input nodes (\emph{driver nodes}) 
that control the (entire) network~\cite{LiuY2011}.
This study has applications in understanding biological networks, communication networks, social networks, electrical circuits, etc.

Structural controllability of systems or networks has been studied in several fields, particularly in control systems theory. 
In a classical work~\cite{LinCT1974} (1974), Lin studied linear time-invariant control systems of the form $(A,b)$: $\dot{x} = Ax + bu$, where matrices $A \in R^{n \times n}$ and $b \in R^n$ are 
time invariant and $x \in R^n$ and $u \in R$, and established that the system $(A,b)$ is structurally controllable if and only if the graph of $(A,b)$ is ``spanned by a cactus".

More recently (2011-) great interest has been generated on the structural controllability of real-world networks~\cite{LiuY2011,Nepusz2012,Cowan2013,Nacher2013}.
Liu Yang et al.~\cite{LiuY2011}, by combining principles of network science with tools from control theory~\cite{Mesbahi2010}, 
studied controllability in gene regulatory, metabolic, social, world-wide web (WWW) and electrical circuit networks.
To identify the \emph{minimum} number $N_D$ of driver nodes required to control the network the authors proposed a maximal-matching based approach -- those nodes that 
are not matched constitute the driver nodes.
Surprisingly, they found that driver nodes tend to avoid hubs in these real-world networks. 
Gene regulatory networks displayed a high $N_D$ indicating that it is necessary to
independently control a large number of genes to fully control the network, while social and WWW networks displayed the smallest $N_D$ indicating that a few individuals could in principle control
the whole network. The former finding is useful for identifying effective drug targets (genes), while the latter is useful to design robust mechanisms to prevent (a few) individuals from bringing down
large social or web networks.

On the other hand, Nepusz and Vicsek~\cite{Nepusz2012} studied controllability from the point of view of edge dynamics, terming it as switchboard dynamics (SBD).
Strikingly different from the conclusions by Liu Yang et al.~\cite{LiuY2011}, under the SBD model, regulatory networks and communications networks were well controllable using only a few driver nodes.
However, Cowan et al.~\cite{Cowan2013} argue that a single time-dependent input is all that is needed for structural controllability, and this input should be applied to the {\sc power dominating set} of the network.
Nacher and Akutsu~\cite{Nacher2013} studied structural controllability of real-world unidirectional bipartite networks.
The authors proposed a variant of the {\sc minimum dominating set} problem to identify driver nodes, and by applying their approach to human drug-target protein networks,
they identified a set of drugs that controlled all protein targets.

While these works consider mostly time-invariant networks, recent studies~\cite{Aswani2009,Chang2012} have proposed the idea of temporal sequence of network
motifs that describe developmental events which cannot be captured by time-invariant models.
However, these works do not specifically focus on network controllability, 
but instead on generating time-variant models that fit the underlying data over time.

Here we study the controllability of \emph{time-variant} networks such as in \emph{cancer}.
From a systems point-of-view,
cancer forms a robust system capable of maintaining stable functioning (cell sustenance and proliferation) despite perturbations~\cite{Kitano2004}. 
Cancer progresses as stages over time typically with increasing aggressiveness and worsening prognosis --
\emph{e.g.} as localised cancer or \emph{in situ}, regional spread, and distant spread or metastasis.
Cancer even of a single organ can be highly diverse, and is therefore studied by categorizing into different subtypes -- 
\emph{e.g.} as basal, luminal-A, luminal-B, HER2+ and normal-like for breast cancer~\cite{Perou2000,Rakha2009}.
Identifying these stages or subtypes and the nodes (driver genes) responsible for \emph{transitions} between them is critical to detect `soft-points' that can break the robustness of cancer, and therefore
aid in developing subtype- or stage-specific anti-cancer therapies.

Differential expression analysis has been traditionally adopted to identify driver genes~\cite{Liang2003,Karanjawala2008}. 
While these analyses manage to capture several ``mountain" genes that show noticeable changes in expression, there are many more ``hills" that often do not display such drastic changes~\cite{Wood2007}. 
These hills are not identifiable through their own behaviour, but their changes are quantifiable when considered in conjunction with other genes; these hills may not be differentially
expressed but are differentially \emph{co-expressed} with other genes~\cite{Hudson2009,Srihari2013}. This is further substantiated in the following case study~\cite{Srihari2013}.

\subsection{An initial analysis}
\label{sec_initial_analysis}

\begin{figure}[!t]
\centering
\includegraphics[scale=0.57]{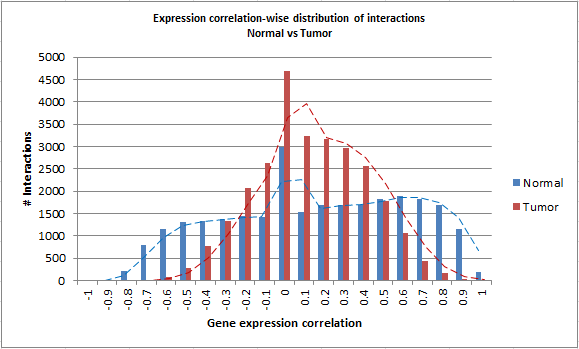}
\caption{Pancreatic normal vs tumour shows significant differences in co-expression patterns among PPIs.}
\label{Prelim_analysis}
\end{figure}

We integrated 29600 high-quality physical interactions among 5824 proteins gathered from Biogrid~\cite{Stark2011} and 39 paired normal and tumour gene-expression samples gathered from a study on
pancreatic ductal adenocarcinoma (PDAC) patients~\cite{Badea2008} to understand differences in behaviour of genes in the tumour \emph{vis-a-vis} normal
(we use the terms \emph{genes} and \emph{proteins} interchangeably). 

We computed the gene expression correlation-wise distribution of interacting gene pairs for normal and tumour conditions
(co-expression is measured as Pearson correlation across samples), as shown in Figure~\ref{Prelim_analysis}.
The gene-expression measurements, although from tissues (mixture of cells) across multiple samples, are from cells with high cellularity, and the figure depicts an `average' picture of
the co-expression pattern in the two conditions.
We observed considerable changes in the correlation of gene pairs in tumour \emph{vis-a-vis} normal -- a reduction in 8701 highly correlated interactions (of absolute correlation $\geq$0.50).
This indicated a potential loss of positively correlated ``accelerators" (interactions driving normal cellular processes) and 
negatively correlated ``brakes" (interactions suppressing tumour inducers and genome instability).
Interestingly, the analysis of ``jumps" (increase or decrease) in correlation revealed two interactions, RBPMS-RHOXF2 and SMN1-TMSB4X, displaying extreme jumps (from +/-[0.9,1] to -/+[0.9,1]).
Among these, RHOXF2, with low expression levels and no noticeable change (mean of 4.67 and 4.34, respectively),
has been implicated as a cancer promoter in pancreatic and gastric cancers~\cite{Shibata-Minoshima2012}.

\begin{figure*}[!t]
\centering
\includegraphics[scale=0.65]{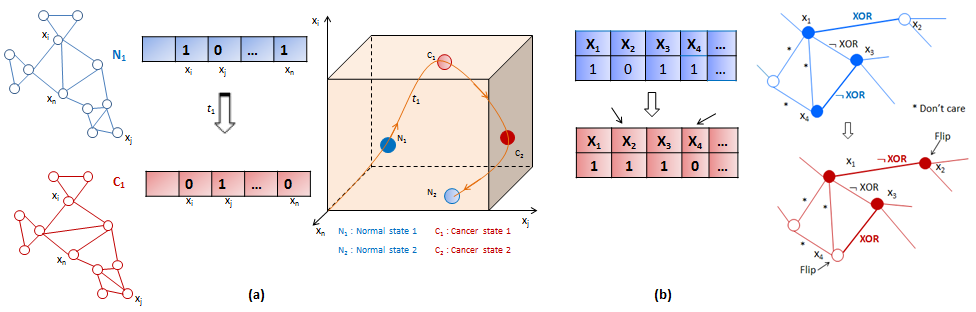}
\caption{BoolSpace: Modeling the `state space' of cancer states over time using Boolean networks.}
\label{fig1}
\end{figure*}

Taking these findings into account, here we hypothesize that changes in gene co-expression patterns, especially among physically interacting protein pairs (PPIs), are strong indicators of
transitions between tumour states. Therefore, we propose a novel model that captures the \emph{dynamics} of tumours based on co-expression patterns of PPI networks across stages,
and use this model to reconstruct the \emph{state space} for tumours.

More specifically, we model the \emph{cancer state space} as a \emph{Boolean state space} wherein
each state is identified by the configuration of a Boolean network that represents the PPI network under a given condition. Each node in the Boolean network is a Boolean variable representing a gene,
and the interactions between nodes are Boolean clauses reflecting co-expression relationships in the PPI network.
Stable states of the network are identified by Boolean satisfying ({\sc sat}) assignments to the nodes, while
transitions between the states are governed by edits to the interactions and corresponding new {\sc sat} assignments to the nodes.
Based on this model, we track the \emph{trajectory} of the Boolean network in the state space to capture progression of the tumour and the genes that drive these transitions (see Figure~\ref{fig1}a).
To identify these genes, we propose an interesting optimization problem called {\sc min flip}, and propose an efficient fixed-parameter tractable algorithm to solve it.
We demonstrate the effectiveness of our model on three case studies involving pancreatic and breast tumours and spinal-cord injury.
We call our model \textbf{BoolSpace}.


\section{Methods}

\subsection{Boolean modeling of cancer state space}

We devise a Boolean model of the cancer state space by integrating PPI network and gene expression profiles from cancer conditions as follows.
Let $H=(V,E)$ be the human PPI network, where $V$ is the set of proteins and $E$ is the set of physical interactions among the proteins.
For each gene (protein) $p \in V$ and any given condition $\Omega$, the gene-expression profile for $p$ consists of expression levels of $p$ measured across multiple samples (\emph{e.g.} patients) in the condition $\Omega$.
Using these expression profiles, for each interacting gene pair $(p,q) \in E$, we measure the co-expression $r(p,q)_{\Omega}$ in $\Omega$. 
Applying a threshold $0 < \delta < 1$ on $r$, we model the interaction $(p,q)$ as a \emph{Boolean clause}: 
\begin{itemize}\itemsep0em
\item if $p$ and $q$ are positively co-expressed, $r(p,q)_{\Omega} \geq \delta$, we model it as $p \BXOR q$ (\emph{i.e.} NOT XOR); and
\item if $p$ and $q$ are negatively co-expressed, $r(p,q)_{\Omega} \leq -\delta$, we model it as $p \XOR q$ (\emph{i.e.} XOR).
\end{itemize}
This results in a \emph{conditional Boolean network} $B_{\Omega} = (V_{\Omega}, E_{\Omega})$, where each $p \in V_{\Omega}$ is a Boolean variable and 
each interaction $(p,q) \in E_{\Omega}$ is a Boolean clause in $p$ and $q$ for condition ${\Omega}$.

When the Boolean clause for the interaction $(p,q)$ evaluates to \texttt{1}, it reflects the co-expression relationship between $p$ and $q$.
Here, $p \BXOR q$ represents the case where both $p$ and $q$ are \texttt{1} or \texttt{0} simultaneously, 
which means both $p$ and $q$ are simultaneously up-regulated or down-regulated, \emph{i.e.} positive co-expression.
On the other hand, $p \XOR q$ represents the case where only one of $p$ or $q$ is \texttt{1} (\texttt{0}) and the other \texttt{0} (\texttt{1}), which means 
only one of $p$ or $q$ is up-regulated while the other is down-regulated, \emph{i.e.} negative co-expression.

The underlying assumption here is that interacting pairs of proteins are likely to be encoded by strongly co-expressed (positive or negative) pairs of genes~\cite{Grigoriev2001,Ge2001}.
Therefore, we consider the generic PPI network as a backbone and condition (contextualize) it using expression profiles to reflect the presence or absence of interactions under different conditions.
If any two genes $p$ and $q$ display strong co-expression ($r(p,q)_{\Omega} \geq \delta$ or $\leq -\delta$) under a condition ${\Omega}$, then we consider the interaction $(p,q)$ to exist in ${\Omega}$,
with the positive or negative co-expression represented by the clauses $\BXOR$ or $\XOR$, respectively, in the Boolean network.

Given $B_{\Omega}$ generated using this model, 
we consider $(p,q)$ to be {\sc satisfied} if we can find a Boolean assignment (\texttt{0/1}) for $p$ and $q$ such that the Boolean clause for $(p,q)$ evaluates to \texttt{1}.
We consider the network $B_{\Omega}$ to be {\sc satisfied}
if we can find a Boolean assignment $\mathcal{B}(B_{\Omega}) = \{b_1,b_2,...,b_n\}$, $b_i = \texttt{0/1}$, spanning all genes $v_i \in V_{\Omega}$ such that every interaction in the network is {\sc satisfied}.
The set of all possible states ({\sc satisfied} as well as {\sc unsatisfied}) a Boolean network can take constitutes its Boolean state space,
where each state is uniquely identified by the configuration and corresponding Boolean assignments for the network.
The {\sc satisfied} states represent stable states because these reflect acceptable 
expression values for genes in the PPI network.

\subsection{Modeling transitions in Boolean space}

We postulate that the Boolean network always transists between {\sc satisfied} states in the Boolean state space.
If the configurations, $B_{\Omega}$ and $B_{\Psi}$, for a network under any two successive conditions ${\Omega}$ and ${\Psi}$ are known, 
we say $B_{\Omega}$ has \emph{transitioned} to $B_{\Psi}$ by \emph{edits} to its interactions.
These edits can be of three types \emph{viz.} loss, gain and `toggling' of interactions, all of which change the configuration of the network.
From condition ${\Omega}$ to ${\Psi}$, an interaction $(p,q)$ is:
\begin{itemize}
\item \emph{lost}, if $r(p,q)_{\Omega} \geq \delta$ or $r(p,q)_{\Omega} \leq -\delta$ but $r(p,q)_{\Psi} \in (-\delta,+\delta)$; 
\item \emph{gained}, if $ r(p,q)_{\Omega} \in (-\delta,\delta)$ but $r(p,q)_{\Psi} \geq \delta$ or $r(p,q)_{\Psi} \leq -\delta$; and
\item \emph{toggled}, if $r(p,q)_{\Omega} \geq \delta$ but $r(p,q)_{\Psi} \leq -\delta$ or \emph{vice versa}.
\end{itemize}
Upon toggling, the Boolean logic on $(p,q)$ changes from $\XOR$ to $\BXOR$ or \emph{vice versa}, and the set of toggled interactions is given by
$\mathcal{T}_{\Omega \Psi} = \{(p,q): p \circ q \in E_{\Omega}, p \bar{\circ} q \in E_{\Psi}; \circ \in \{\XOR,\BXOR\}  \}$
(recollect ``jumps" in co-expression mentioned under `Initial analysis').
The total set of interactions edited is represented as $\mathcal{E}_{\Omega \Psi}$.
These edits capture changes in co-expression patterns among interacting gene pairs, and
therefore transitions in the Boolean space reflect `rewiring' of the PPI network between conditions.
Based on this model, we are now interested in identifying the genes driving these transitions of the network.

\subsubsection{Deducing drivers of state transitions}
Given a satisfying assignment $\mathcal{B}(B_{\Omega})$,
we hypothesize that the \emph{minimum} subset of genes to be \emph{flipped} (from \texttt{0} to \texttt{1} or \emph{vice versa}) 
to maintain the network {\sc satisfied} upon transit to $B_{\Psi}$ constitutes the genes driving this transition.
To identify these driver genes, we propose the following problem:
\begin{quote}
{\sc Min Flip:} Given the network $B_{\Omega} = (V_{\Omega},E_{\Omega})$ and 
its satisfying assignment $\mathcal{B}(B_{\Omega})$ for a condition ${\Omega}$, and the set of edited interactions $\mathcal{E}_{\Omega \Psi}$ relative to another condition ${\Psi}$,
find a minimal subset of genes $V'_{\Omega} \subseteq V_{\Omega}$ to be flipped such that $B_{\Omega}$ remains {\sc satisfied} when $\mathcal{E}_{\Omega \Psi}$ is edited.
\end{quote}
Note that we \emph{edit} or \emph{toggle} interactions but \emph{flip} genes.
For example, in Figure~\ref{fig1}b, the interactions $(x_1,x_2)$ and $(x_3,x_4)$ have toggled from $\XOR$ to $\BXOR$ and $\BXOR$ to $\XOR$, respectively, and
to resatisfy this network, we flip $x_2$ and $x_4$.


\subsection{Parameterizing {\sc Min Flip}}
In the {\sc Min Flip} formulation above, we need to know the initial {\sc sat} assignment $\mathcal{B}(B_{\Omega})$ to identify the flipped genes.
In an $n$-gene network with only $\XOR$ or $\BXOR$ clauses there are polynomial (in $n$) and
in a general network there are potentially $O(2^n)$~\cite{Valiant1979} number of {\sc sat} assignments to choose as our initial assignment. 
Here we always select the assignment with the minimum number of \texttt{1}'s as our initial assignment $\mathcal{B}(B_{\Omega})$.

In a network with only $\XOR$ or $\BXOR$ clauses an assignment with the minimum number of \texttt{1}'s (called the {\sc Min-Ones-2sat} problem) is 
determinable in polynomial time, and therefore {\sc Min Flip} is solvable in polynomial time (shown later).
On the other hand, {\sc Min Flip} is equivalent (details skipped here) to the 
{\sc Min-Ones-2sat}, which is NP-complete in a general network~\cite{Valiant1979,Mishra2010}.
Therefore, to solve {\sc Min Flip} in general, we assume a bound on the flipped genes and present a tractable algorithm relative to this bound.

We present a \emph{fixed-parameter tractable} (FPT) algorithm for {\sc Min Flip}
{\em parameterizing} on the number of flipped genes. For an input of size $n$,
FPT algorithms run in $O(f(k).n^c)$ time, where $k$ is a positive integer (the parameter), $f$ a (typically exponential) function dependent only on $k$, and $c$ is a constant independent of $k$~\cite{Nied2006}.
FPT algorithms, in many cases, are more practical than the na\"{i}ve $O(n^k)$ algorithms when $k$ is ``small enough" \cite{Nied2006,Srihari2008}.
A classical example is of the {\sc vertex cover} problem, for which a number of FPT algorithms exist in the literature parameterizing primarily on the
size of the vertex cover, the best one achieving an asymptotic running time of $1.2738^k.n^{O(1)}$~\cite{Chen2006} (for an introduction to FPT algorithms, refer to \cite{Nied2006}).

We reformulate {\sc Min Flip} relative to a parameter $k > 0$ as follows:
\begin{quote}
{\sc $k$-Flip:} Given the network $B_{\Omega} = (V_{\Omega},E_{\Omega})$, 
its satisfying assignment $\mathcal{B}(B_{\Omega})$ for a condition ${\Omega}$, and the set of edited interactions $\mathcal{E}_{{\Omega \Psi}}$ relative to a condition ${\Psi}$, 
find the subset of genes $V'_{\Omega} \subseteq V_{\Omega}$, $|V'_{\Omega}| \leq k$, to be flipped such that $B_{\Omega}$ remains {\sc satisfied} when $\mathcal{E}_{{\Omega \Psi}}$ is edited.
\end{quote}
We expect $k << |V_{\Omega}|$.


\subsection{Solving {\sc Min Flip}}

We first state some preliminaries.
For a gene $p$ in network $B_{\Omega}$, $N_{\Omega}(p)$ is the set of neighbors and $E_{\Omega}(p)$ is the set of incident interactions of $p$. 
The subsets of {\sc satisfied} and {\sc unsatisfied} interactions, $E^{S}_{\Omega}(p)$ and $E^{U}_{\Omega}(p)$ respectively, form a partition of $E_{\Omega}(p)$, that is,
$E^{S}_{\Omega}(p) \cup E^{U}_{\Omega}(p) = E_{\Omega}(p)$ and $E^{S}_{\Omega}(p) \cap E^{U}_{\Omega}(p) = \emptyset$.

\begin{lemma} 
For a gene $p$, if $|E^{U}_{\Omega}(p)| > k$ then $p$ belongs to the final solution $F$ of flipped genes.
\end{lemma}
\begin{proof}
If $p \notin F$ then, each of its neighbors $N_{\Omega}(p)$ need to be flipped at the very least to satisfy $E^{U}_{\Omega}(p)$. 
However, by doing so, we overshoot $F$ i.e., $|F| > k$.
\end{proof}


\subsubsection{An FPT algorithm for general networks}
We propose an FPT algorithm similar to that known for the {\sc vertex cover} problem~\cite{Nied2006}.
The inputs to the algorithm are the network $B_{\Omega}$ in condition ${\Omega}$, a {\sc satisfying} assignment $\mathcal{B}(B_{\Omega})$,
the edited subset $\mathcal{E}_{\Omega \Psi}$ relative to a condition ${\Psi}$, and $k > 0$.

\noindent \emph{Pre-processing:}
We perform the edits $\mathcal{E}_{\Omega \Psi}$ in $B_{\Omega}$.
At each step in our algorithm we maintain two partitions of $E_{\Omega}$:
(i) $U$ of all {\sc unsatisfied} interactions, initially $U := \mathcal{E}_{\Omega \Psi}$; and
(ii) $S$ of all {\sc satisfied} interactions, initially $S := E_{\Omega} \setminus U$.

We repeatedly find genes $p$ such that $|E^{U}_{\Omega}(p)| > k$ and do $F := F \cup \{p\}$ (by Lemma 1). 
For all interactions $(p,q)$ that get {\sc satisfied}, we do $S := S \cup \{(p,q)\}$ and $U := U \setminus \{(p,q)\}$.
At the end of this step, the resultant network should have at most $k.(k-|F|)$ {\sc unsatisfied} interactions, if it is to have a solution.
This is because for each gene $p \notin F$, $|E^{U}_{\Omega}(p)| \leq k$, and at most $k - |F|$ of these can be flipped and added to $F$, which can satisfy
at most $k.(k-|F|)$ interactions.
If $|U| > k.(k-|F|)$ we return a {\sc no}, else we set $k' := k - |F|$ and 
continue with the following recursive search.
\vspace{1.5mm}

\noindent \emph{Bounded search} (see Algorithm~\ref{Alg1}):
At every step of the recursive search we pick an interaction $(p,q)$ and branch on the following two cases: we either flip $p$ or flip $q$.
We recursively solve the problem by this two-way branching until we have flipped $k'$ genes or have found a solution.
Upon flipping $p$ (or $q$), we set $F := F \cup \{p\}$ (or $F := F \cup \{q\}$) and decrement $k'$ by 1. For all interactions $(p,x)$ (or $(q,x)$) that are incident on
$p$ (or $q$) and are {\sc satisfied} by the flip, we set $S := S \cup \{(p,x) \}$ and $U := U \setminus \{(p,x) \}$ (or $S := S \cup \{(q,x) \}$ and $U := U \setminus \{(q,x) \}$).
At any step if $k'=0$ and $U \neq \emptyset$, we return a {\sc no}, else we return an {\sc yes} along with $F$.

Since we perform a two-way branching at every recursive step and upto a depth of at most $k'$,
the total number of nodes in the search tree is at most $2^{k'}$, and
because we spend at most a polynomial time (in $|E_{\Omega}|$) at each of these nodes, total the running time is bounded in the worst case by $O^*(2^k)$,
\emph{i.e.} FPT.

\emph{Lazy speed-up:}
We can speed-up the above algorithm in certain cases (\emph{e.g.} when the Boolean clauses are of the form $p \wedge q$) by making the following observation:
if $(p,q)$ remains {\sc unsatisfied} upon flipping $p$, then the only way to satisfy $(p,q)$ is to flip $q$ as well, and 
therefore we can perform the operations of two recursive calls within one call based on the satisfiability of $(p,q)$.
Consequently, in any step after flipping $p$, if $(p,q)$ remains {\sc unsatisfied}, then
instead of performing a call immediately, we delay the call to post flipping of $q$. 
We then decrement $k'$ by 2, and therefore speed-up the descent down the tree and also avoid the overhead of a function call.



\begin{algorithm}
\caption{\texttt{bool} $k$-Flip$(U, S , F,k)$}
\begin{algorithmic}
\STATE \small 
\STATE \texttt{bool} $r$;
\IF{$k=0$ and $U \neq \emptyset$}
\STATE return~{\sc FALSE};
\ENDIF
\STATE
\STATE Pick $(p,q)$;			\emph{// Pick a random interaction.}
\STATE ---------------
\STATE Flip $p$; $F := F \cup p$;				
\IF{$(p,q)$ is {\sc satisfied} }
\STATE $U := U \setminus (p,q)$, $S := S \cup (p,q)$;

\IF{$U \neq \emptyset$ and $k > 0$} 
\STATE \emph{//Decrement k by 1 and recurse.}
\STATE $r$ := $k$-Flip$(U, S, F, k-1)$;		
\ENDIF
\IF{$r ==$ TRUE}
\STATE return TRUE and $F$;
\ENDIF

\ENDIF
\STATE ---------------
\STATE Flip $q$; $F := F \cup q$
\STATE $U := U \setminus (p,q)$, $S := S \cup (p,q)$;

\IF{$U \neq \emptyset$ and $k > 1$}
\STATE \emph{//Decrement k by 2 and recurse.}
\STATE $r$ := $k$-Flip$(U, S, F, k-2)$;	
\ENDIF
\IF{$r =$ TRUE}
\STATE return TRUE and $F$;
\ELSE
\STATE return FALSE;
\ENDIF
\STATE
\STATE \textbf{end} $k$-Flip;
\end{algorithmic}
\label{Alg1}
\end{algorithm}


\subsubsection{Initial assignment for general networks}
The problem of determining an initial assignment with the minimum number of \texttt{1}'s, called the {\sc Min-Ones 2-Sat} problem, is NP-complete in a general network~\cite{Valiant1979,Mishra2010}. 
Therefore, to identify the initial assignment, we parameterize the problem as follows:
\begin{quote}
{\sc $k$-Ones 2-Sat}:
Given a Boolean network $B_{\Omega}$ and a parameter $k>0$, find a {\sc sat} assignment $\mathcal{B}(B_{\Omega})$ such that $\mathcal{B}(B_{\Omega})$ has at most $k$ \texttt{1}'s.
\end{quote}
Observe here that {\sc $k$-Ones 2-Sat} is equivalent to {\sc $k$-Flip} by starting with an all-\texttt{0} assignment. Therefore,
to find the solution $\mathcal{B}(B_{\Omega})$, we just reset every gene to \texttt{0} and run Algorithm~\ref{Alg1} with the parameter as $k$. 
The number of \texttt{0}'s flipped (at most $k$) is the solution to {\sc $k$-Ones 2-Sat}, determinable in $O(2^k)$ time, giving us the initial assignment $\mathcal{B}(B_{\Omega})$.


\subsubsection{A polynomial-time algorithm for $\XOR/\BXOR$-networks}
We first show that in a network with only $\XOR/\BXOR$ clauses, there are only a polynomial number of satisfiability assignments.
\begin{theorem}
The number of satisfiability assignments for a Boolean network $B$ containing only $\XOR/\BXOR$ clauses is twice the number of
components of $B$.
\end{theorem}
\begin{proof}
We construct a subnetwork $B'$ using only the $\XOR$-interactions of $B$. If $B'$ is satisfiable, then we should be able to 2-colour each of its components, that is,
assign a \texttt{1/0} to each gene such that no two genes have the same assignment. This is equivalent to finding whether $B'$ is bipartite, and can be done in two ways for each of the components.
Next, we pick each remaining $\BXOR$-interaction and add it to $B'$. If an interaction $(p,q)$ is incident on a gene $p$ already present in $B'$, then
$q$ should have the same assignment as $p$, else this interaction belongs to a new component and there are two ways of satisfying it. Therefore, the total number of ways of satisfying $B$
is twice the number of components in $B$.
\end{proof}

We next give a polynomial-time algorithm for {\sc Min Flip} in $\XOR/\BXOR$-networks.
For a given such network $B_{\Omega}$, there are only a polynomial number of {\sc sat} assignments (Theorem 1), and therefore we can identify the initial 
{\sc sat} assignment $\mathcal{B}(B_{\Omega})$ with the minimum number of \texttt{1}'s by simply checking each of these assignments, in polynomial time.

Observe that among the interactions in $\mathcal{E}_{\Omega\Psi}$, the lost interactions do not change the satisfiability of the network, while for the gained or toggled interactions $(p,q) \in \mathcal{E}_{\Omega\Psi}$ 
we need to flip only one of $p$ or $q$ to resatisfy $(p,q)$.
Therefore, there are at most $2.|\mathcal{E}_{\Omega\Psi}|$ ways to resatisfy the network upon editing $\mathcal{E}_{\Omega\Psi}$, and
we can identify the assignment achievable using the minimum number of flips in polynomial time.

\subsection{Practical considerations}

\subsubsection{Network structure}
The network structure might not always allow a satisfying assignment. Therefore, in practice,
we allow at most a certain (small) number of interactions to be left {\sc unsatisfied} in our solution. 
This number is specified as an input to our algorithm (here, 10\% of the total interactions).

\subsubsection{Contradictory cycles}
Cycles in the network that cause contradictory assignments can interfere with our search for solutions. 
Consider a cycle $C = \{p,q,...,r,p\}$ in an $\XOR/\BXOR$-network. Starting at $p$ and assigning it a \texttt{0(1)}, if we go around the cycle and arrive at a contradictory
assignment \texttt{1(0)} for $p$, we call $C$ a contradictory cycle. We overcome such cycles in the network by arbitrarily marking an interaction in each of the cycles to be
left {\sc unsatisfied} in the network.


\section{Results}

We implemented BoolSpace using C/C++ on an Intel Core i5 Linux machine. The source codes are available at: \url{http://www.bioinformatics.org.au/tools-data}.
Although the networks considered here contain only $\XOR/\BXOR$-interactions, we employed the algorithm for general networks in our experiments.

\subsection{Preparation of experimental data}
We applied BoolSpace on three case studies:
(i) pancreatic normal and tumour conditions in human;
(ii) BRCA1 and BRCA2 breast tumours in human; and
(iii) across five time-points after spinal-cord injury (SCI) in rats. 
While the third case study is not from cancer, much of the regeneration mechanisms post-injury involve progressive stages similar to cancer.
We gathered the following datasets for our experiments.

\emph{PPI datasets:}
We gathered \emph{Homo sapiens}, \emph{Mus musculus} and \emph{Rattus norvegicus} PPI data inferred from multiple low- and high-throughput experiments deposited in Biogrid v3.1.93~\cite{Stark2011}.
To minimize false-positives in these datasets~\cite{Srihari2013b} we used a scoring scheme, 
Iterative-CD (with 30 iterations) by Liu Guimei et al.~\cite{Liu2009}, to assign a reliability score for each interaction in the PPI networks.
The score (between 0 and 1) reflects the reliability of interactions by accounting for the number of common neighbors shared among the proteins in each pair. 
Discarding low-scoring interactions ($<$0.20) resulted in a high-quality human PPI network of 29600 interactions among 5824 proteins (average node degree $d_{avg}$ = 10.16), and 
a mammalian (rat and mouse) PPI network of 3215 interactions among 1146 proteins ($d_{avg}$ = 5.61).

\emph{Gene expression datasets:}
The pancreatic ductal adenocarcinoma (PDAC) gene-expression datasets  were gathered from the studies by Badea \emph{et al.}~\cite{Badea2008}, containing of 39 matched pairs (78 total) of normal and tumour samples
(GEO GSE15471).
The breast expression profiles came from the study on familial BRCA1 and BRCA2 tumours by Waddell \emph{et al.}~\cite{Waddell2010},
containing 19 BRCA1- and 30 BRCA2-tumour samples (GEO GSE19177).
The rat spinal-cord injury (SCI) datasets came from the study by De Baise \emph{et al.}~\cite{DeBaise2005}, containing samples from five time-points post SCI: 0 hours, 4 hours,
72 hours, 7 days and 28 days with at least 15 samples per time-point (ArrayExpress E-GEOD-5296).
In all cases, the original processed (normalized) datasets released by the studies were used.

\emph{Some background on these case studies:}
PDAC accounts for most ($95$\%) pancreatic tumours and is predominantly characterized by dysfunctioning (by mutation) of the KRAS oncogene and of the CDKN2A, SMAD4 and TP53 tumour-suppressor genes~\cite{Jones2008}.

On the other hand, breast tumours are very heterogeneous, and extensive gene expression profiling studies have classified sporadic tumours into clinically relevant molecular subtypes \emph{viz.}
luminal A, luminal B, triple-negative/basal-like, HER2+ and normal-like~\cite{Perou2000,Rakha2009}.
Most breast tumours are luminal and they tend to be estrogen-receptor positive
(ER$+$) and/or progesterone-receptor positive (PR$+$). Luminal tumours have relatively better prognosis and survival rates. 
Triple-negative tumours are characterised by lack of ER (ER$-$), PR (PR$-$) and HER2 (HER2$-$) expression. These tumours are highly aggressive relative to the luminal subtypes and are associated with high recurrence,
distant metastasis and poor survival. Basal-like tumours form a subtype of triple-negative tumours that stain positive for EGFR/HER1 and express high-molecular-weight form of cytokeratine 5/6~\cite{Rakha2009}.
The breast expression profiles we employ here come from the study on familial BRCA1 and BRCA2 tumours (that have germline BRCA1/BRCA2 mutations) by Waddell \emph{et al.}~\cite{Waddell2010}.
BRCA1 tumours are known to be predominantly triple-negative/basal-like while BRCA2 tumours predominantly luminal~\cite{Lakhani1998}.

SCI causes secondary biochemical changes which are typically associated with 
hemorrhage, metabolic failure, inflammatory/immune activation, loss of ionic homeostasis, lipid degradation, production of free radicals, and neurotransmitter/neuromodulator imbalances~\cite{DeBaise2005,Giovanni2003}.
Such alterations contribute to death of neurons and oligodendroglial cells, glial proliferation, demyelination, and axonal loss~\cite{DeBaise2005}.

\subsection{Setting the parameter $k$}
The parameter $k$ determines the size of the allowable set of genes to be flipped. 
While there is no standard procedure to choose $k$, we would like a $k$ that is as close as possible to the minimum number of flipped genes (the minimum is unknown to us). 
To determine such a $k$, we provide a rule-of-thumb to be used \emph{in practice}. 
This rule is based on the observation that typically when $k$ is much farther from the minimum, the FPT algorithm tends to takes lesser time,
compared to when $k$ is closer to the minimum. This is because the search is depth-first in nature and therefore, 
with a larger $k$ it is easier to find a deep path containing a solution quickly 
(by including the first-available $k$ genes into the solution) instead of exploring the rest of the search tree and trying for a smaller solution.
Although this ``quick" solution is of size at most $k$ and is correct, we would like to force the algorithm to explore other (potentially smaller) solutions, if achievable.
Therefore, our rule-of-thumb works as follows: we start with $k = |V_{\Omega}| - 1$, and repeatedly decrement $k$ until we can find a solution at each iteration within ``reasonable" time $T$
(here, we set $T=100$ seconds). If a solution is found within $T$ time, we consider the algorithm is not exploring the search tree sufficiently, 
and therefore we continue decrementing $k$. We stop at the $k$ at which the search takes more than $T$ time.


\subsection{Analysis of network in different conditions}

\begin{table*}[!t]
\caption{Transition of Boolean networks between conditions in three case studies}
\label{Table1}
\centering
\begin{tabular}{| l || l | | c || c | c  c  c || c  c  c || c || c || c || c |}
\hline
{}		&{}			&{}		& {}		& \multicolumn{3}{c||}{\#Interactions}				& \multicolumn{3}{c||}{\#Edits}								& {Parameter}		& {\#Genes}		&{Running} 		\\
{Case study}	&{Transition}		&{$\delta$}	& {\#Genes}	& {Total}	& {$\XOR$}	&{$\BXOR$}			& {Lost}	& {Gained}	&{Toggled}						& {$k$}			&{flipped}		&{time (sec)$^*$}	\\
\hline

		&			&		&		&		&		&				&		&		&							&			&			&			\\
		&			&0.80		&	1174	&	1701	&	241	&	1460			&	1672	&	16	&	0						&	10		&	9		& 6			\\

Pancreatic	&Normal {\em to}	&0.75		&	1712	&	2896	&	573	&	2323			&	2836	&	40	&	4						&	25		&	23		& 10			\\

		&tumour			&0.70		&	2265	&	4300	&	1056	&	3244			&	4185	&	95	&	4						&	60		&	54		& 13			\\
\hline
		&			&		&		&		&		&				&		&		&							&			&			&			\\

		&			&0.80		&	270	&	302	&	106	&	196			&	293	&	23	&	0						&	5		&	1		& 8			\\

Breast		&BRCA1 {\em to}		&0.75		&	604	&	646	&	227	&	419			&	620	&	45	&	2						&	15		&	11		& 10			\\

		&BRCA2			&0.70		&	1090	&	1170	&	373	&	797			&	1116	&	95	&	4						&	50		&	46		& 10			\\
\hline
		&			&		&		&		&		&				&		&		&							&			&			&			\\

		&			&0.80		&	25	&	15	&	0	&	15			&	4	&	15	&	0						&	5		&	0		& 1			\\

Spinal		&0hr {\em to} 4hr 	&0.75		&	35	&	22	&	0	&	22			&	9	&	28	&	0						&	5		&	1		& 1			\\

		&			&0.70		&	42	&	26	&	0	&	26			&	9	&	45	&	0						&	20		&	14		& 1			\\

\cdashline{2-13}
		&			&		&		&		&		&				&		&		&							&			&			&			\\
		&			&0.80		&	108	&	87	&	3	&	73			&	15	&	76	&	0						&	5		&	3		& 1			\\

cord		&4hr {\em to} 72hr 	&0.75		&	66	&	41	&	4	&	37			&	24	&	93	&	0						&	5		&	4		& 1			\\

		&			&0.70		&	99	&	62	&	6	&	56			&	38	&	130	&	1						&	25		&	23		& 1			\\

\cdashline{2-13}
		&			&		&		&		&		&				&		&		&							&			&			&			\\
		&			&0.80		&	107	&	87	&	3	&	84			&	39	&	39	&	0						&	5		&	1		& 1			\\

injury		&72hr {\em to} 7d 	&0.75		&	136	&	112	&	4	&	108			&	49	&	46	&	0						&	5		&	2		& 1			\\

		&			&0.70		&	185	&	154	&	8	&	146			&	75	&	46	&	0						&	10		&	6		& 1			\\

\cdashline{2-13}
		&			&		&		&		&		&				&		&		&							&			&			&			\\
		&			&0.80		&	108	&	87	&	1	&	86			&	42	&	22	&	0						&	5		&	4		& 1			\\

		&7d {\em to} 28d 	&0.75		&	131	&	109	&	2	&	107			&	45	&	33	&	0						&	10		&	6		& 1			\\

		&			&0.70		&	153	&	126	&	5	&	121			&	53	&	49	&	0						&	25		&	22		& 1			\\

\cline{2-13}
		&			&		&		&		&		&				&		&		&							&			&			&			\\
		&			&0.80		&	25	&	15	&	0	&	15			&	4	&	56	&	0						&	5		&	5		& 1			\\

		&0hr {\em to} 28d 	&0.75		&	35	&	22	&	0	&	22			&	8	&	83	&	0						&	10		&	7		& 1			\\

		&			&0.70		&	42	&	26	&	0	&	26			&	11	&	107	&	0						&	20		&	16		& 1			\\

\hline
\multicolumn{13}{l}{\scriptsize *Includes the time for finding initial Boolean assignment and the solution after edits.}
\end{tabular}
\end{table*}


Table~\ref{Table1} shows properties of the Boolean network and the number of genes flipped while it transists between different conditions for
$\delta = \{0.80,0.75,0.70\}$ in the three case studies -- pancreatic and breast tumours and spinal-cord injury.
The number of $\BXOR$ interactions are higher than $\XOR$ in these networks indicating higher number of
positively co-expressed interacting pairs compared to negatively co-expressed; this is not surprising since we expect higher number of ``accelerator" interactions compared to ``brakes", 
and has been observed in several previous studies as well~\cite{ChuLH2008}. 
As the $\delta$-threshold decreases, we observe an increase in the network sizes because we allow for lowly co-expressed gene pairs.
This also leads to higher number of edits in terms of lost, gained and toggled interactions between the conditions. 

The correlation-wise distributions for interactions before and after the edits showed significant differences:
KS test -- Normal vs PDAC $D_{NP} = 23.12 > K_{\alpha = 0.05}$; BRCA1 vs BRCA2 $D_{B12} = 22.85 > K_{\alpha = 0.05}$; and SCI
between 7hr and 7d $D_{7hr-7d} = 17.03 > K_{\alpha = 0.05}$, where  $K_{\alpha = 0.05} = 1.36$.

While it is not entirely surprising to see (given our initial analysis in Section I) a large number of edited (particularly lost) interactions between normal and tumour (here, normal and PDAC), the noticeably large number of
interactions edited between two subtypes of the same cancer (here, BRCA1 and BRCA2 tumours) is very interesting. This strongly suggests considerable differences in PPI wiring between the two breast tumours.
In general, BRCA1 tumours have higher number of interactions compared to BRCA2 tumours. Whether this is reflective of the higher aggressiveness of BRCA1 tumours~\cite{Lakhani1998} is interesting to explore.

Further, while there were higher number of total edited interactions from normal to tumour compared to BRCA1 tumour to BRCA2 tumour, the gained interactions from BRCA1 tumour to BRCA2 tumour were higher than the
gained interactions from normal to tumour. Even though the two cancers (pancreatic and breast) are not directly comparable, but this trend indicates that during transition from normal to tumour,
we predominantly see a weakening of the cellular machinery (as loss in interactions), but between subtypes, we can expect considerable rewiring involving not only a loss
but also gain of interactions. This extensive rewiring might be the cause of considerable differences between the two tumour subtypes.

In the case of SCI, the number of gained interactions between 0hr to 72hr is higher than the lost, but between 72hr to 28d the number of lost interactions is higher than gained. Whether this is indicative of a 
pattern of response to the injury is worth further exploration -- for example,
a considerable number of new interactions are formed during the initial stages to aid recovery, and subsequently lost when the recovery stabilizes during the final stages.

\subsubsection{Functional analysis of edited interactions}

DAVID-based (\url{http://david.abcc.ncifcrf.gov/})~\cite{Dennis2003} functional analysis of the edited interactions in pancreatic and breast showed significant enrichment ($p<0.01$) for
Biological Process (BP) terms \emph{viz.} Cell cycle, Chromatin organization,
DNA repair and RNA splicing, indicating considerable rewiring in core cellular processes responsible for genome stability and maintenance. For example, interactions involving the tumour suppressors
TP53 and SMAD4 in pancreatic tumour, and those involved in DNA double-strand break repair namely BRE and BRCC3 apart from BRCA1, BRCA2 and TP53 in breast tumours showed significant decrease in correlations
indicating loss of interactions.
Among the interactions edited in spinal-cord injury, we noted significant enrichment ($p<0.01$) for MAPK signalling, TGF-$\beta$ signalling, Inflammatory response, Cell proliferation and Apoptosis pathways.
This indicated activation of regenerative mechanisms including response to inflammation and growth-factor pathway actuation for regeneration of cells.


\subsection{Analysis of driver genes}

Next, we collated the flipped genes (Tables~\ref{Table2a} and \ref{Table2b}) and studied them using differential expression and functional analysis.

\subsubsection{Differential expression of flipped genes}

\begin{figure}[!t]
\centering
\includegraphics[scale=0.72]{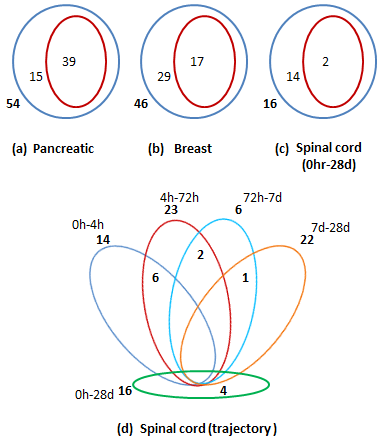}
\caption{Analysis of flip genes: (a)-(c): Differentially expressed genes (red) among the flipped genes (blue); (d) Genes flipped at each stage of SCI.}
\label{fig4}
\end{figure}

We assessed our flipped genes using differential expression analysis ($p$-value $<0.001$), as shown in Figure~\ref{fig4} (a)-(c).
Interestingly, while many of the flipped genes were also differentially expressed, there were several others which were not captured by the analysis.
Investigation into these genes showed that these directly or indirectly (through one or two neighbors) interacted in the PPI network with 
key genes implicated in pancreatic and breast tumours. In other words, these were differentially \emph{co-expressed} and belonged to
the same pathways as the key genes.

\subsubsection{Functional analysis of flipped genes}
Table~\ref{Table3} shows the top GO terms (using DAVID~\cite{Dennis2003}) enriched for the flipped genes in the three case studies. 
For the spinal-cord study, we show the enrichment only for genes flipped between the two extreme conditions (0hr to 28days).

The pancreatic genes were involved in Cell cycle, Wnt signalling and Mismatch repair pathways, which have been implicated in pancreatic tumours~\cite{Jones2008}.
The high enrichment for Neurotrophin signalling further the nexus between neural genes and pancreatic carcinogenesis~\cite{Biankin2012, Srihari2013}.
The breast genes were enriched for Homologous recombination, which is a key pathway in DNA double-strand break repair and houses the two breast-cancer susceptibility genes, BRCA1 and BRCA2.
The SCI genes were enriched for Immune response and Growth-factor signalling pathways indicating activation of regenerative mechanisms.

\begin{table*}[!t]
\caption{Genes flipped between tumour states in human}
\label{Table2a}
\centering
\scriptsize
\begin{tabular}{| c || l l l || l l l |}
\hline
{Transition}	& \multicolumn{3}{|c||}{Normal {\em to} PDAC}		& \multicolumn{3}{|c|}{BRCA1 {\em to} BRCA2}	\\
\hline		
	 	& Brca1		& Jun		& Ruvbl1		& Brca1		& Ppar$\gamma$	& Sp1		\\
Flipped		& Csnk2b	& Krt15		& Sfn			& Esr1		& Tp53		& Hsf1		\\
genes		& Fgfr		& Mcm5		& Usp10			& Cebp$\beta$	& Myb		&		\\
		& Fos		& Psmd1		&			& Gata1		& Foxa1		&		\\
		& Hras		& Rbx1		&			& Gata3		& Fos		&		\\
\hline
\multicolumn{7}{l}{\scriptsize Genes shown here have degree $\geq 5$}\\
\end{tabular}
\end{table*}

\begin{table*}[!t]
\caption{Genes flipped at different stage-transitions post spinal-cord injury in rats}
\label{Table2b}
\centering
\scriptsize
\begin{tabular}{| c || l l | l l l | l | l l l || l l l |}
\hline
{Transition}	&\multicolumn{2}{|c|}{0hr {\em to} 4hr}	&	\multicolumn{3}{c|}{4hr {\em to} 72hr}			&	{72hr {\em to} 7d}		&	\multicolumn{3}{c ||}{7d {\em to} 28d}		&	\multicolumn{3}{ c |}{0hr {\em to} 28d}					\\
\hline
	&Angpt2 &	Tnfrsf1b		&Pten			&Tnfrsf1b	& Smad4				&	Ppar$\gamma$c1a			 &Pten				&Csk		&	Atm	&	Hdac1		  &	Csk	  &	Bcl10		  \\
	&Sparc  &	Mapk1			&Angpt2			&Akt1		& Fabp5				&	Sp1				 &Hdac1				&Ccng1		&	Mapk3	&	Ppar$\gamma$c1a	  &	Ccng1	  &	Nf$\kappa$bia		 \\
Flipped	&Cdc14  &	Jak2			&Cflar			&Bmpr1a		& Neurod1			&	Akt1				 &Cflar				&Ppp1ca		&	Casp9 	&	Ccnd3		  &	Ppp1ca	  &	Chek2		 \\
genes	&Il1r1  &	Relb 			&Hoxa3			&Csk		& Atm				&	Csk				 &Sp1				&Smad1		&	Bcl10	&	Cdk4		  &	Egfr	  &	Casp9 		 \\
	&Bmp4	&	Tlr2			&Cd14			&Pms2		& Tgfbr1			&	Eif4g2				 &Ccnd3				&Smad4		&	Nf$\kappa$bia	&	akt1	  &	Mapk3	  &	Cdkn1a		 \\
	&Myd88  &	Nf$\kappa$bia		&Il1r1			&Ppp1ca		& Trib3				&	Zeb1				 &Cdk4				&Egfr		&	Chek2	&	Traf2		  &		  &			 \\
	&Wnt4   &	Bcl3			&Myd88			&Eif4e		& Tlr2				&					 &Akt1				&Hif1a		&	Cdkn1a	&			  &		  &			 \\
	&	&				&Hfe			&Smad1  	&				&					 &Traf2				 &		 &	       &			  &		   &			  \\
	
\hline
\end{tabular}
\end{table*}

\begin{table*}[!t]
\caption{Enrichment for top Gene Ontology terms in flipped genes}
\label{Table3}
\centering
\scriptsize
\begin{tabular}{| l || l | c | c || l | c | c || l | c | c |}
\hline
{}		& \multicolumn{9}{c|}{Case study}																																			\\
{GO}		&	\multicolumn{3}{c||}{Pancreatic}						&	\multicolumn{3}{c||}{Breast}			&		\multicolumn{3}{c|}{Spinal cord injury}															\\
\hline
		&	{Term}				& Genes			&	$p$-value		&	{Term}					& Genes			&	$p$-value		&	{Term}						& Genes			&	$p$-value		\\		
		&					& (\%)			&				&						& (\%)			&				&							& (\%)			&				\\
\hline
		&	Cell cycle			&	4.6		&	3.5(-13)	&	Cell cycle					&	3.2		&	2.7(-07)		&	Apoptosis					&	21.7		&	1.3(-04)		\\
		&	Neurotrophin signal.		&	3.0		&	1.7(-05)	&	Nucleotide excision rep.			&	1.6		&	1.5(-05)		&	TGF-$\beta$ sig.				&	17.4		&	2.3(-03)		\\
		&	Nucleotide excision rep.	&	1.7		&	1.9(-05)	&	DNA repli.					&	1.4		&	6.4(-05)		&	Toll-like receptor				&	17.4		&	3.4(-03)		\\
KEGG		&	Pancreatic cancer		&	2.1		&	5.7(-05)	&	Adipocytokine signal.				&	1.8		&	7.5(-07)		&	Pancreatic cancer				&	13.0		&	2.1(-02)		\\
pathways	&	Adipocytokine signal.		&	2.0		&	9.7(-04)	&	Apoptosis					&	2.1		&	1.2(-04)		&	colourectal cancer				&	13.0		&	2.9(-02)		\\
		&	Regulation of autophagy		&	1.3		&	3.4(-04)	&	Homologous recomb.				&	1.0		&	1.6(-03)		&	MAPK signal.					&	17.4		&	4.8(-02)		\\
		&	Mismatch rep.			&	1.0		&	5.2(-04)	&	Insulin signal.					&	2.2		&	6.0(-03)		&							&			&				\\
		&	Wnt signal.			&	2.8		&	2.2(-03)	&	Mismatch rep.					&	0.9		&	2.8(-03)		&							&			&				\\
		
\hline
		&	Cell cycle			&	17.3		&	1.6(-35)	&	Chromosome org.					&	14.3		&	1.5(-43)		&	Enzyme-receptor signal.				&	34.8		&	1.6(-07)		\\	
Biological	&	Chromosome org			&	13.0		&	6.2(-33)	&	Chromatin mod.					&	12.2		&	1.3(-40)		&	Serine/threonine kinase				&	21.7		&	6.8(-06)		\\
Process		&	Chromatin mod.			&	8.9		&	1.0(-27)	&	Transcription reg.				&	31.6		&	1.1(-24)		&	Inflammatory res.				&	26.1		&	2.5(-05)		\\
		&					&			&			&							&			&				&	Defense/immune res.				&	30.4		&	7.8(-05)		\\
		&					&			&			&							&			&				&	Cell proliferation 				&	30.4		&	1.6(-04)		\\
\hline
\end{tabular}
\end{table*}

Table~\ref{Table2b} and Figure~\ref{fig4}d show overlaps among the flipped genes at each transition post SCI from 0hr till 28d.
For example, 14 genes were flipped from 0hr to 4hr and 23 genes were flipped from 4hr to 72hr stages with 6 genes in common.
Interestingly, the overlaps between successive stages were not considerable ($<50\%$) indicating that 
sets of genes involved in different cellular processes were flipped at each transition. For example,
the genes flipped during the initial stages (0hr to 4hr) were predominantly enriched for immune response and the proteins were localized in extra-cellular matrix and membranes, 
while those during the final stages (7d to 28d) were predominantly enriched for cell apoptosis, growth and proliferation, and were localized in the nucleus (Figure~\ref{fig6}).
This suggests a pattern to SCI response --
activation of immune response during the initial stages, and regeneration through cell apoptosis, growth and proliferation during the final stages.
Further, the analysis also highlights that genes belonging to cell cycle progression are involved in neuronal responses to DNA damage and/or cell stress after SCI, as also observed in 
earlier studies~\cite{Giovanni2003}. For example,
Pten (O08586) is a tumour suppressor which modulates cell cycle progression and cell survival, and is involved in
controlling the rate of newborn neuron-integration during adult neurogenesis, including correct neuron positioning, dendritic development and synapse formation.

\begin{figure*}[!t]
\centering
\includegraphics[scale=0.60]{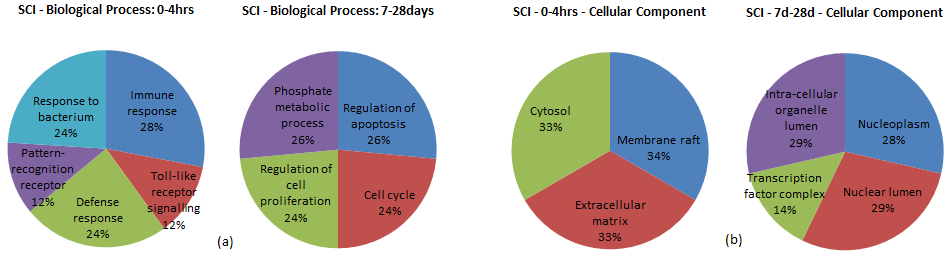}
\caption{Distribution of flipped genes in SCI at 0-4hrs and 7-28days for (a) Biological Process and (b) Cellular Component terms.}
\label{fig6}
\end{figure*}


\subsubsection{In-depth study of some flipped genes}

Several of the flipped genes were cyclin-dependent kinases (CDKs), particularly 
the serine-threonine kinases that act as ``ON/OFF" switches and play crucial roles in the regulation of cell proliferation, apoptosis and cell differentiation;
the flipping of genes in our Boolean model might possibly be related to these cellular switching mechanisms.
For example, we noticed flipping of Ccnd3 (P30282), a member of the G1/S-specific cyclin D3-CDK4 complex that phosphorylates and inhibits members of the retinoblastoma (RB) 
protein family including RB1 and regulates the cell-cycle during G1/S transition.
It also acts as a substrate for SMAD3 (a tumour suppressor), phosphorylating SMAD3 in a cell-cycle-dependent manner and repressing its transcriptional activity
(\url{http://www.uniprot.org/uniprot/P30282}~\cite{UniProt2013}).

Among the flipped genes were also a few transcription factors (TFs).
For example, the following TFs flipped between BRCA1 and BRCA2 tumours: GATA3, ESR1, FOXA1 and XBP1.
These four TFs are ER targets, and BRCA1 tumours are ER$-$ and therefore are likely to show lower expression of ER targets compared to
BRCA2 tumours, which are express ER$+$~\cite{Lakhani1998}.

Finally, we also noticed striking overlaps between the genes and/or pathways enriched in pancreatic tumour and SCI. For example, 
Pten (O08586), Myd88 (P22366), Wnt4 (P22724), Tnfrsf1b (P25119), Atm (Q62388), Bcl3 (Q9Z2F6) and Jak2 (Q62120) are involved in TGF-$\beta$, Wnt and JAK-STAT signalling and have been
implicated in pancreatic tumours~\cite{Jones2008}. This supports the close nexus between pancreatic tumourigenesis and neuronal response and development~\cite{Biankin2012}.


\section{Discussion}

\subsection{Why \emph{minimum} gene flips makes sense}
We argue using a simple yet intuitive example why we select the minimum number of genes (instead of, say, the maximum) to be flipped to determine driver genes.
Consider a gene $t$ (say, a transcription factor) that interacts with $m$ genes, $\{g_1,g_2,...,g_m\}$ (its targets), in the network under condition $\Omega$. 
Now suppose that a change in the expression level of $t$ (and not of the $m$ genes)
results in the interactions $\mathcal{E}_{tg} = \{(t,g_1),(t,g_2),...,(t,g_m)\}$ becoming UNSATISFIED upon transit to condition $\Psi$.
To resatisfy $\mathcal{E}_{tg}$, we could either flip $t$ or each of the $m$ genes.
However, in this case, flipping the maximum set of genes (the $m$ genes) instead of the minimum (only $t$) identifies the incorrect set of driver genes.
Therefore, by flipping the minimum set, we always attempt to identify the genes that are ``more" responsible for driving the transition.

Note that selecting the minimum set tends to favor hubs. Therefore, our model agrees more with Nepusz and Viscek~\cite{Nepusz2012} and Nacher and Akutsu~\cite{Nacher2013}
than with Liu Yang et al.~\cite{LiuY2011}. Since many of the hubs in PPI networks correspond to essential proteins~\cite{Kang2010}, and because many of these hubs that we
found were CDKs that act as biological ``ON/OFF" switches, it is possible that our flipped genes are indeed important proteins involved in rewiring of the PPI network.

\subsection{Cancer robustness partly stems from `passing of the baton' between genes}

Although the experiments presented in this work are still preliminary, based on our findings (Figures~\ref{fig4}d and \ref{fig6}) we hypothesize that robustness of cancer partly
stems from the fact that genes from different
biological processes and/or cellular components are involved in different stages (timepoints) during tumour progression.
As a result of this constant ``passing of the baton" between the genes, tumours can evade therapy if the genes that are targeted at a particular timepoint 
are no longer driving the tumour (i.e. have passed on the baton to other genes) or are not yet involved in the tumour (i.e. not yet received the baton) at that timepoint.

Having said that, there is a certain sequence in which genes are involved in the tumour, and therefore deciphering this sequence will be crucial to develop
effective anti-cancer therapies. Applying BoolSpace, we can identify the genes driving the tumour at different stages during tumour progression, and by identifying a ``cover set" of these genes 
(e.g. the cover set contains at least one gene from each transition) that can be simultaneously targeted, we should be able to
break the robustness of the tumour.

\subsection{A distance measure between tumour stages}
It is not hard to see that the (minimum) number of genes flipped between stages is a \emph{metric} because it essentially is the Hamming distance between Boolean vectors for the stages.
Therefore, the idea of using the minimum number of genes flipped as a `distance' measure between tumour stages in the Boolean state space, in which stages that are more (biologically) similar are
placed closer in the state space compared to stages that are less (biologically) similar, is worth further exploration.
It is interesting to check if this distance captures (biological) differences between tumours or tumour stages.

\section{Conclusion}
Cancer forms a robust system by maintaining stable functioning (cell proliferation and sustenance) despite perturbations (e.g. drug targeting)~\cite{Kitano2004}.
Inherent to this robustness is the continuous progression or change in system characteristics so as to constantly evade system failure inflicted through perturbations.
Therefore, identifying genes driving this progression is critical to develop effective anti-cancer therapies.

In this work, we have proposed a novel model called BoolSpace to track the progression of cancer in a Boolean state space.
In this state space, a Boolean network, constructed by integrating PPI and gene-expression datasets, transits between Boolean satisfiability states by 
editing interactions and flipping genes. We hypothesize that the minimum number of genes flipped in response to edits in interactions corresponds to the genes driving these transitions.
To identify these flipped genes, we propose an optimization problem called {\sc min flip} and a fixed-parameter tractable algorithm to solve the problem efficiently.
Experiments on three case studies -- pancreatic and breast tumours in human and spinal-cord injury in rats -- suggest that many of the identified genes are involved in tumourigenic activity.
Several of these genes are serine/threonine kinases that act as biological ``ON/OFF switches" within cells and are involved in key cell cycle, proliferation, apoptosis and differentiation processes.
Finally, we hypothesize that cancer robustness partly stems from ``passing of the baton" between genes responsible for driving different stages of the tumour, and therefore an effective therapy should likely
target a ``cover set" of genes across a succession of stages to break the robustness of cancer.


\subsection* {\em Acknowledgments}
We thank Dr Ashish Anand (IIT G) for valuable discussions, and the 
anonymous reviewers for their valuable suggestions.
Funding: SS is supported under an Australian National Health and Medical Research Council (NHMRC) grant 1028742 to Dr Peter T. Simpson and MAR.


\ifCLASSOPTIONcaptionsoff
  \newpage
\fi




\begin{thebibliography}{1}






\bibitem{LiuY2011}
Y. Liu, J. Stoline, A. Barabasi,
\newblock {``Controllability of complex networks"},
\newblock {\em Nature}, vol. 473, no. 12 pp. 167--173, 2011.


\bibitem{LinCT1974}
C.T. Lin,
\newblock {``Structural controllability"},
\newblock {\em IEEE Transactions on Automatic Control}, vol. 19 pp. 201--208, 1974.



\bibitem{Mesbahi2010}
M. Mesbahi, M. Egerstedt,
\newblock {``Graph Theoretic Methods in Multiagent Networks"}, 
\newblock {\em Princeton University Press}, 2010.


\bibitem{Nepusz2012}
T. Nepusz, T. Vicsek,
\newblock {``Controlling edge dynamics in complex networks"},
\newblock {\em Nature Physics}, vol. 8 pp. 568--773, 2012.




\bibitem{Cowan2013}
N.J. Cowan, E.J. Chastain, D.A. Vilhena, J.S. Freudenberg, C.T. Bergstrom,
\newblock  {``Nodal dynamics, not degree distributions, determine the structural controllability of complex networks"},
\newblock {\em PLoS ONE}, vol. 7 no. 6 pp. e38398, 2012.


\bibitem{Nacher2013}
J.C. Nacher, T. Akutsu,
\newblock {``Structural controllability of unidirectional bipartite networks"},
\newblock {\em Scientific Reports}, vol. 3 pp. 1647, 2013.



\bibitem{Aswani2009}
A. Aswani, N. Boyd, C. Tomlin,
\newblock {``Graph-theoretic topological control of biological genetic networks"},
\newblock {\em In proceedings of American Control Conference}, pp. 1700--1705, 2009.





\bibitem{Chang2012}
Y.H. Chang, J. Gray, C. Tomlin,
\newblock {``Optimization-based inference for temporally evolving networks with applications in biology"},
\newblock {\em Journal of Computational Biology}, vol. 19 no. 12 pp. 1307--1323, 2012.




\bibitem{Kitano2004}
H. Kitano,
\newblock {``Cancer as a robust system: implications for anticancer therapy"},
\newblock {\em Nature Reviews Cancer}, vol. 4 pp. 227--235, 2004.


\bibitem{Perou2000}
C.M. Perou, T. Sorlie, M.B. Eisen, M. van de Rijn, SS. Jeffrey, C.A. Rees, J.R. Pollack, D.T. Ross, H. Johnsen, L.A. Akslen, O. Fluge, A. Pergamenschikov, C. Williams, S.X. Zhu, P.E. Lonning, A.L. Borresen-Dale, P.O. Brown, D. Botstein,
\newblock {``Molecular portraits of human breast tumours"},
\emph{Nature}, vol. 406 no. 6797 pp. 747, 2000.


\bibitem{Rakha2009}
E.A. Rakha, S.E. Elsheik, M.A. Aleskandarany, H.O. Habashi, A.R. Green, D.G. Powe, M.E. El-Sayed, A. Benhasouna, J.S. Brunet, L.A. Akslen, A.J. Evans, R. Blamey, J.S. Reis-Filho, W.D. Foulkes, I.O. Ellis,
\newblock{``Triple-negative breast cancer: distinguishing between basal and nonbasal subtypes"},
\emph{Clinical Cancer Research}, vol. 15 no. 7 pp. 2302--2310, 2008.



\bibitem{Liang2003}
P. Liang, A. Bardee,
\newblock {``Analysing differential gene expression in cancer"},
\newblock {\em Nature Reviews Cancer}, vol. 3 pp. 869--876, 2003.






\bibitem{Karanjawala2008}
Z.E. Karanjawala, P.B. Illei, R. Ashfaq, J.R. Infante, K. Murphy, A. Pandey, R. Schulick, J. Winter, R. Sharma, A. Maitra, M. Goggins, R.H. Hruban,
\newblock {``New markers of pancreatic cancer identified through differential gene expression analyses: claudin 18 and annexin A8"},
\newblock {\em American Journal of Surgical Pathology}, vol. 32 no. 2 pp. 188--196, 2008.



\bibitem{Wood2007}
L.D. Wood, W.D. Parsons, S. Jones, J. Lin, T. Sjoblom, R.J. Leary, D. Shen, S.M. Boca, T. Barber, J. Ptak, N. Silliman, S. Szabo, Z. Dezso, V. Ustyanksky, T. Nikolskaya, Y. Nikolsky, R. Karchin, P.A. Wilson, J.S. Kaminker, Z. Zhang, R. Croshaw, J. Willis, D. Dawson, M. Shipitsin, J.K. Willson, S. Sukumar, K. Polyak, B.H. Park, C.L. Pethiyagoda, P.V. Pant, D.G. Ballinger, A.B. Sparks, J. Hartigan, D.R. Smith, E. Suh, N. Papadopoulos, P. Buckhaults, S.D. Markowitz, G. Parmigiani, K.W. Kinzler, V.E. Velculescu, B. Vogelstein,
\newblock {``The genomic landscapes of human breast and colourectal cancers"},
\newblock \textit{Science}, vol. 318 no. 5853 pp. 1108--1113, 2007.



\bibitem{Hudson2009}
N.J. Hudson, A. Reverter, B.P. Dalrymple,
\newblock  {``A differential wiring analysis of expression data correctly identifies the gene containing the causal mutation"},
\newblock {\em PLoS Computational Biology}, vol. 5 no. 5 pp. e1000382, 2009.


\bibitem{Srihari2013}
S. Srihari, M.A. Ragan, 
\newblock {``Systematic tracking of dysregulated modules identifies novel genes in cancer"},
\textit{Bioinformatics}, vol. 29 no. 12 pp. 1553-61, 2013.






\bibitem{Stark2011}
C. Stark, B.J. Breitkreutz, A. Chatr-Aryamontri, L. Boucher, R. Oughtred, M.S. Livstone, J. Nixon, K. Van Auken, X. Wang, X. Shi, T. Reguly, JM. Rust, A. Winter, K. Dolinski, M. Tyers,
\newblock {``The BioGRID interaction database: 2011 update"},
\newblock {\em Nucleic Acids Research}, vol. 39 pp. D698--D704, 2011.


\bibitem{Badea2008}
L. Badea, V. Herlea, S.O. Dima, T. Dumitrascu, I. Popescu,
\newblock {``Combined gene expression analysis of whole-tissue and microdissected pancreatic ductal adenocarcinoma identifies gene specifically overexpressed in tumour epithelia"},
\newblock {\em Hepatogastroenterology}, vol. 55 pp. 2015--2026, 2008.



\bibitem{Shibata-Minoshima2012}
F. Shibata-Minoshima, T. Oki, N. Doki, F. Nakahara, S. Kageyama, J. Kitaura, J. Fukuoka, T. Kitamura,
\newblock {``RHOXF2 (PEPP2) as a cancer-promoting gene by expression cloning"},
\newblock {\em International Journal of Oncology}, vol. 40 no. 1 pp. 93--98, 2012.







\bibitem{Grigoriev2001}
A. Grigoriev,
\newblock {``A relationship between gene expression and protein interactions on the proteome scale: analysis of the bacteriophage T7 and the yeast Saccharomyces cerevisiae"},
\newblock {\em Nucleic Acids Research}, vol. 29 pp. 3513--3519, 2001.


\bibitem{Ge2001}
H. Ge, Z. Liu, G.M. Church, M. Vidal,
\newblock {``A relationship between gene expression and protein interactions on the proteome scale: analysis of the bacteriophage T7 and the yeast Saccharomyces cerevisiae"},
\newblock {\em Nature Genetics}, vol. 29 pp. 482--486, 2001.


\bibitem{Valiant1979}
L.G. Valiant,
\newblock{``The complexity of enumeration and reliability problems"},
\newblock {\em SIAM Journal on Computing} vol. 8 no. 3 pp. 410--421, 1979.


\bibitem{Mishra2010}
N. Mishra, N.S. Narayanaswamy, V. Raman, B.S. Shankar,
\newblock{``Solving {\sc minones-2 sat} as a fast vertex cover"},
\newblock {\em In Proceedings of Mathematical Foundations of Computer Science}, pp. 549--555, 2010.


\bibitem{Nied2006} 
R. Niedermeier, 
\newblock {``Invitation to Fixed-Parameter Algorithms"},
{\em Oxford University Press}, 2006.


\bibitem{Srihari2008}
S. Srihari, H.K. Ng, K. Ning, H.W. Leong,
\newblock {``Detecting hubs and quasi cliques in scale-free networks"},
\newblock {\em In Proceedings of International Conference on Pattern Recognition}, pp. 1-4, 2008.


\bibitem{Chen2006}
J. Chen, I.A. Kanj, G. Xia,
\newblock {``Improved parameterized upper bounds for vertex cover"},
\newblock {\em In Proceedings of Mathematical Foundations of Computer Science}, Springer LNCS vol. 4162 pp. 238--249, 2006.







\bibitem{Srihari2013b}
S. Srihari, H.W. Leong,
\newblock {``A survey of computational methods for protein complex prediction from protein interaction networks"},
\newblock {\em Journal of Bioinformatics and Computational Biology}, vol. 11 no. 2 pp. 1230002, 2013.


\bibitem{Liu2009}
G. Liu, L.S. Wong, H.N. Chua,
\newblock {``Complex discovery from weigthed PPI networks"},
\newblock {\em Bioinformatics}, vol. 25 pp.1891--1897, 2009.




\bibitem{Waddell2010}
N. Waddell, J. Arnold, S. Cocciardi, L. da Silva, A. Marsh, J. Riley, C.N. Johnstone, M. Orloff, G. Assie, C. Eng, L. Reid, P. Keith, M. Yan, S. Fox, P. Devilee, A.K. Godwin, F.B. Hogervorst, F. Couch, kConFab Investigators, S. Grimmond, J.M. Flanagan, K.K. Khanna, P.T. Simpson, S.R. Lakhani, G. Chenevix-Trench,
\newblock {Subtypes of familial breast tumours revealed by expression and copy number profiling}.
\newblock {\em Breast Cancer Research and Treatment}, vol. 123 pp. 661--667, 2010.






\bibitem{DeBaise2005}
A. De Baise, S.M. Knoblach, S.D. Giovanni, C. Fan, A. Molon, E.P. Hoffman, A.I. Faden,
\newblock {``Gene expression profiling of experimental traumatic spinal cord injury as a function of distance from impact site and injury severity"},
\newblock {\em Physiological Genomics}, vol. 22 no. 3 pp. 368--381, 2005.





\bibitem{Jones2008}
S. Jones, X. Zhang, W.D. Parsons, J.C. Lin, R.J. Leary, P. Angenendt, P. Mankoo, H. Carter, H. Kamiyama, A. Jimeno, S.M. Hong, B. Fu, M.T. Lin, E.S. Calhoun, M. Kamiyama, K. Walter, T. Nikolskaya, Y. Nikolsky, J. Hartigan, D.R. Smith, M. Hidalgo, S.D. Leach, A.P. Klein, E.M. Jaffee, M. Goggins, A. Maitra, C. Iacobuzio-Donahue, J.R. Eshleman, S.E. Kern, R.H. Hruban, R. Karchin, N. Papadopoulos, G. Parmigiani, B. Vogelstein, V.E. Velculescu, K.W. Kinzler,
\newblock{``Core signalling pathways in human pancreatic cancers revealed by global genomic analysis"},
\newblock{\em Science}, vol. 321 pp. 1801--1806, 2008.



\bibitem{Lakhani1998}
S.R. Lakhani, J. Jacquemier, J.P. Sloane, B.A. Gusterson, T.J. Anderson, M.J. van de Vijver, L.M. Farid, D. Venter, A. Antoniou, A. Storfer-Isser, E. Smyth, C.M. Steel, N. Haites, R.J. Scott, D. Goldgar, S. Neuhausen, P.A. Daly, W. Ormiston, R. McManus, S. Scherneck, B.A. Ponder, D. Ford, J. Peto, D. Stoppa-Lyonnet, Y.J. Bignon, J.P. Struewing, N.K. Spurr, D.T. Bishop, J.G. Klijn, P. Devilee, C.J. Cornelisse, C. Lasset, G. Lenoir, R.B. Barkardottir, V. Egilsson, U. Hamann, J. Chang-Claude, H. Sobol, B. Weber, M.R. Stratton, D.F. Easton,
\newblock {``Multifactorial analysis of differences between sporadic breast cancers and cancers involving BRCA1 and BRCA2 mutations"},
\newblock{\em Journal of the National Cancer Institute}, vol. 90 no. 15 pp. 1138--1145, 1998.


\bibitem{Giovanni2003}
D.S. Giovanni, S.M. Knolbach, C. Brandoli, S.A. Aden, E.P. Hoffman, A.I. Faden,
\newblock{``Gene profiling in spinal cord injury shows role of cell cycle in neuronal death"},
\newblock{\em Annals of Neurology}, vol. 53, no. 4 pp. 454--468, 2003.





\bibitem{ChuLH2008}
L.H. Chu, B.S. Chen, 
\newblock {``Construction of a cancer-perturbed protein-protein interaction network for discovery of apoptosis drug targets"},
\newblock{\em BMC Systems Biololgy}, vol. 2 no. 56, 2008.




\bibitem{Dennis2003}
G. Dennis, B.T. Sherman, D.A. Hosack, J. Yang, W. Gao, H.C. Lane, R.A. Lempicki,
\newblock{``DAVID: Database for Annotation, Visualization, and Integrated Discovery"},
\newblock{\em Genome Biology}, vol. 4 pp. R60, 2003.






\bibitem{Biankin2012}
A.V. Biankin, N. Waddell, K.S. Kassahn, M.C. Gingras, L.B. Muthuswamy, A.L. Johns, D.K. Miller, P.J. Wilson, A.M. Patch, J. Wu, D.K. Chang, M.J. Cowley, B.B. Gardiner, S. Song, I. Harliwong, S. Idrisoglu, C. Nourse, E. Nourbakhsh, S. Manning, S. Wani, M. Gongora, M. Pajic, C.J. Scarlett, A.J. Gill, A.V. Pinho, I. Rooman, M. Anderson, O. Holmes, C. Leonard, D. Taylor, S. Wood, Q. Xu, K. Nones, J.L. Fink, A. Christ, T. Bruxner, N. Cloonan, G. Kolle, F. Newell, M. Pinese, RS. Mead, J.L. Humphris, W. Kaplan, M.D. Jones, E.K. Colvin, A.M. Nagrial, E.S. Humphrey, A. Chou, V.T. Chin, L.A. Chantrill, A. Mawson, J.S. Samra, J.G. Kench, J.A. Lovell, R.J. Daly, N.D. Merrett, C. Toon, K. Epari, N.Q. Nguyen, A. Barbour, N. Zeps, Australian Pancreatic Cancer Genome Initiative, N. Kakkar, F. Zhao, Y.Q. Wu, M. Wang, D.M. Muzny, W.E. Fisher, F.C. Brunicardi, S.E. Hodges, J.G. Reid, J. Drummond, K. Chang, Y. Han, L.R. Lewis, H. Dinh, C.J. Buhay, T. Beck, L. Timms, M. Sam, K. Begley, A. Brown, D. Pai, A. Panchal, N. Buchner, R. De Borja, R.E. Denroche, C.K. Yung, S. Serra, N. Onetto, D. Mukhopadhyay, M.S. Tsao, P.A. Shaw, GM. Petersen, S. Gallinger, R.H. Hruban, A. Maitra, C.A. Iacobuzio-Donahue, R.D. Schulick, C.L. Wolfgang, R.A. Morgan, RT. Lawlor, P. Capelli, V. Corbo, M. Scardoni, G. Tortora, M.A. Tempero, K.M. Mann, N.A. Jenkins, P.A. Perez-Mancera, D.J. Adams, D.A. Largaespada, L.F. Wessels, A.G. Rust, L.D. Stein, D.A. Tuveson, N.G. Copeland, E.A. Musgrove, A. Scarpa, J.R. Eshleman, T.J. Hudson, R.L. Sutherland, D.A. Wheeler, J.V. Pearson, J.D. McPherson, R.A. Gibbs, S.M. Grimmond,
\newblock{``Pancreatic cancer genomes reveal aberrations in axon guidance pathway genes"},
\newblock{\em Nature}, vol. 491 no. 7424 pp. 399-405, 2012.



\bibitem{UniProt2013}
The UniProt Consortium,
\newblock{``Update on activities at the Universal Protein Resource (UniProt) in 2013"},
\newblock{\em Nucleic Acids Research}, vol. 41 pp. D43--D47, 2013.


\bibitem{Kang2010}
K. Ning, H.K. Ng, S. Srihari, H.W. Leong, A. Nesvizhskii,
\newblock{``Examination of the relationship between essential genes in PPI network and hub proteins in reverse nearest neighbor topology"},
\newblock{\em BMC Bioinformatics}, vol. 11 no.515, 2010.












\end{thebibliography}
%

%

\begin{IEEEbiography}{Sriganesh Srihari}
Sriganesh Srihari is a Research Officer (post-doctoral fellow) at the Institute for Molecular Bioscience, the University of Queensland, Australia.
He received his PhD from the National University of Singapore, his MSc from Nanyang Technological University Singapore and his BTech from
National Institute of Technology India, all in computer science.
His research interests include computational biology and bioinformatics, systems biology, data mining and databases.
\end{IEEEbiography}

\begin{IEEEbiography}{Venkatesh Raman}
Venkatesh Raman did his PhD in computer science at the University of
Waterloo, Canada and since then he has been a faculty at the Institute of
Mathematical Sciences, Chennai. His research interests include space efficient data structures and
parameterized complexity. Homepage: \url{http://www.imsc.res.in/~vraman/}
\end{IEEEbiography}

\begin{IEEEbiography}{Hon Wai Leong}
Hon Wai Leong is an Associate Professor in the Department of Computer Science at the National University of Singapore. 
He received the B.Sc. (Hon) degree in Mathematics from the University of Malaya and the Ph.D. degree in Computer Science from the University of Illinois at Urbana-Champaign. 
His research interest is in the design of optimization algorithms for problems from diverse application areas including VLSI-CAD, transportation logistics, multimedia systems, and 
computational biology. In computational biology, his current interests includes computational proteomics, fragment assembly, comparative genomics, and analysis of PPI networks. 
He has a passion for nurturing young talents and gives many workshop on creative problem solving and computational thinking. 
In 1992, he started the Singapore training program for the IOI (International Olympiad in Informatics). 
He is a member of ACM, IEEE, ISCB, and a Fellow of the Singapore Computer Society.
Homepage: \url{http://www.comp.nus.edu.sg/~leonghw/}
\end{IEEEbiography}

\begin{IEEEbiography}{Mark A. Ragan}
Mark A. Ragan received his Bachelor degree in Biochemistry from the University of Chicago, USA and his PhD degree in Biology from Dalhousie University, Halifax, Canada. He is currently at The University of Queensland, where he is founding Head of Genomics and Computational Biology at the Institute for Molecular Bioscience, and Affiliate Professor in the School of Information Technology and Electrical Engineering. He is also Director of the Australian Research Council (ARC) Centre of Excellence in Bioinformatics. His research areas include comparative and computational genomics and systems biology.
\end{IEEEbiography}

\end{document}